\documentclass [journal,onecolumn,11pt]{IEEEtran}
\usepackage{amsfonts,amsmath,amssymb}
\usepackage{indentfirst, setspace}
\usepackage{url,float}
\usepackage[american]{babel}
\usepackage[T1]{fontenc}
\usepackage[utf8]{inputenc}
\usepackage{color}
\usepackage[a4paper,ignoreall]{geometry}
\usepackage{longtable}
\usepackage{graphicx}
\usepackage[a4paper,ignoreall]{geometry}
\usepackage{multicol}
\usepackage{stfloats}
\usepackage{enumerate}
\usepackage{amsmath}
\usepackage{amsfonts}
\usepackage{amssymb}
\usepackage{algorithm}
\usepackage[noend]{algpseudocode}
\usepackage{cite}
\usepackage{amsthm}
\usepackage{booktabs}
\usepackage{mathrsfs}
\usepackage{multirow}
\usepackage{microtype}
\usepackage{amssymb}
\usepackage{verbatim}
\usepackage{accsupp}
\usepackage{listings}
\usepackage{cases}
\usepackage{xcolor}
\usepackage[overload]{empheq}
\def\qu#1 {\fbox {\footnote {\ }}\ \footnotetext { From Qu: {\color{red}#1}}}
\def\hqu#1 {}
\def\kq#1 {\fbox {\footnote {\ }}\ \footnotetext { From KangQuan: {\color{blue}#1}}}
\def\hkq#1 {}

\newtheorem{Th}{Theorem}
\newtheorem{Cor}{Corollary}
\newtheorem{Prop}{Proposition}
\newtheorem{Prob}{Problem}

\newtheorem{Def}{Definition}
\newtheorem{example}{Example}

\newtheorem{Rem}{Remark}

\newcommand{\tr}{{\rm Tr}}

\newcommand{\AC}{\ensuremath{\mathsf{AC}}}

\newcommand{\gf}{{\mathbb F}}

\newcommand{\Z}{\mathbb {Z}}

\newcommand{\supp}{{\rm Supp}}

\usepackage{todonotes}

\newcommand{\mkq}[1]{{{\color{blue}#1}}}

\makeatletter
\newcommand{\figcaption}{\def\@captype{figure}\caption}
\newcommand{\tabcaption}{\def\@captype{table}\caption}
\makeatother

\begin{document}
	\title{
		On the Differential-Linear Connectivity Table \\ of Vectorial Boolean Functions}
	\author{{Anne Canteaut, Lukas Kölsch, Chao Li, Chunlei Li, \\ Kangquan Li,  Longjiang Qu and Friedrich Wiemer}
	\thanks{\noindent Anne Canteaut is with Inria, Paris, France. Lukas Kölsch is with  University of Rostock, Germany.
		Chunlei Li is with the Department of Informatics, University of Bergen, Bergen N-5020, Norway.
		Kangquan Li, Longjiang Qu and Chao Li are with the College of Liberal Arts and Sciences,
		National University of Defense Technology, Changsha, 410073, China.
		Longjiang Qu is also with the State Key Laboratory of Cryptology, Beijing, 100878, China. Friedrich Wiemer is with the  Horst G{\"{o}}rtz Institute for IT Security, Ruhr University Bochum, Germany.
		\textbf{Emails}: anne.canteaut@inria.fr, lukas.koelsch@uni-rostock.de,   lichao\_nudt@sina.com, 	 chunlei.li@uib.no,
		 likangquan11@nudt.edu.cn, 
		 ljqu\_happy@hotmail.com, friedrich.wiemer@rub.de
	}
}

	\maketitle
	
\begin{abstract}

Vectorial Boolean functions are crucial building-blocks in symmetric ciphers. 
Different known attacks on block ciphers have resulted in diverse cryptographic criteria for vectorial Boolean functions,
such as differential uniformity and nonlinearity.  
Very recently, Bar-On et al. introduced at Eurocrypt'19 a new tool, called 
the differential-linear connectivity table (DLCT), which allows for taking
into account the dependency between the two subciphers $E_0$ and $E_1$ involved in differential-linear attacks. This new notion leads to significant improvements of differential-linear attacks on several ciphers.
This paper presents a theoretical characterization of the DLCT of vectorial
Boolean functions and also investigates this new criterion for some families of functions with specific forms.

More precisely, we firstly reveal the connection between the DLCT and the autocorrelation of vectorial Boolean functions,
we characterize properties of the DLCT by means of the Walsh transform of the function and of its differential distribution table, and we present generic bounds on the highest magnitude occurring in the DLCT of 
vectorial Boolean functions, which coincides (up to a factor~\(2\)) with the well-established notion of absolute indicator.  
Next, we investigate the invariance property of the DLCT of vectorial Boolean functions under the affine, 
extended-affine, and Carlet-Charpin-Zinoviev (CCZ) equivalence and exhaust the DLCT spectra of optimal $4$-bit S-boxes under affine equivalence. 
Furthermore, we study the DLCT of APN, plateaued and AB functions and establish its connection with other cryptographic criteria.
 Finally, we investigate the DLCT and the absolute indicator of some specific polynomials with optimal or low differential uniformity, including monomials, cubic functions, quadratic functions and inverses of quadratic permutations.

\end{abstract}	


\section{Introduction}
{
  Let $n, m$ be two arbitrary positive integers. We denote by $\gf_{2^n}$ the finite field with $2^n$ elements and by $\gf_2^n$ the $n$-dimensional vector space over $\gf_2$. 
  Vectorial Boolean functions from $\gf_2^n$ to $\gf_{2}^m$, also  called  $(n,m)$-functions, play a crucial role in block ciphers.  
Many attacks have been proposed against 
block ciphers, and have led to diverse criteria, such as low differential uniformity, high nonlinearity, high algebraic degree, etc, that the implemented cryptographic functions must satisfy.  
In Eurocrypt'18, Cid et al. \cite{EC:CHPSS18} introduced a new concept on S-boxes: the boomerang connectivity table (BCT) that similarly
analyzes the dependency between the upper part and lower part of a block cipher in a boomerang attack.
The work of \cite{EC:CHPSS18} quickly attracted attention in the study of BCT property of cryptographic functions \cite{ToSC:BouCan18,LQSL2019,MTX2019,ToSC:SonQinHu19}
and stimulated research progress in other cryptanalysis methods.
  Very recently, in Eurocrypt'19, Bar-On et al. \cite{EC:BDKW19} introduced  a new tool called the differential-linear connectivity table (DLCT) that
analyzes the dependency between the two subciphers in differential-linear attacks, thereby improving the efficiency of the attacks introduced in~\cite{C:LanHel94}.  
The authors of \cite{EC:BDKW19} also presented the relation between the DLCT and the differential distribution table (DDT) of S-boxes. 

This paper aims to provide a theoretical characterization of the main properties of the DLCT, explicitly of the set formed by all its entries and of the highest magnitude in this set, for generic vectorial Boolean functions. 
To this end, we firstly show that the DLCT coincides (up to a factor~\(2\)) with the autocorrelation of vectorial Boolean functions, 
	which is extended from  Boolean functions.
Based on the study of the autocorrelation  of vectorial Boolean functions, we give some characterizations of the DLCT by means of the Walsh transform and the DDT, and provide a lower bound on 
the absolute indicator (i.e., equivalently, on the highest absolute value in the DLCT excluding the first row and first column) of any $(n,m)$-function; then we 
exhibit an interesting divisibility property of the autocorrelation of $(n,m)$-functions $F$, which implies that the entries of DLCT of any $(n, n)$-permutations are divisible by $4$.  
 Next, we investigate the invariance property of the autocorrelation (and the DLCT) of vectorial Boolean functions under affine, extended-affine (EA) and Carlet-Charpin-Zinoviev (CCZ) equivalence, and 
show that the autocorrelation spectrum is affine-invariant and its maximum magnitude is EA-invariant but not CCZ-invariant.
 Based on the classification of optimal $4$-bit S-boxes by Leander and Poschmann \cite{LP2007}, 
 we explicitly calculate their autocorrelation spectra (see Table \ref{Tab-4Autocorrelation}).  
 Moreover, for certain functions like APN, plateaued and AB functions, we present the relation of  their autocorrelation (and DLCT)  with other cryptographic criteria. We show that the autocorrelation of APN and AB/plateaued functions can be converted to the Walsh transform of two classes of balanced Boolean functions. 
 Finally, we investigate the autocorrelation spectra of some special polynomials with optimal or low differential uniformity, including monomials, cubic functions, quadratic functions and inverses of quadratic permutations.  

 The rest of this paper is organized as follows. Section \ref{Preliminaries} recalls basic definitions, particularly the new notion of DLCT, the generalized notion of autocorrelation, and the connection between them. Most notably, we show that the highest magnitude in the DLCT coincides (up to a factor~\(2)\) with the absolute indicator of the function. Section \ref{Characterizations} is devoted to the characterization of the autocorrelation: we firstly characterize the autocorrelation by means of the Walsh transform and of the DDT of the function. We then exhibit generic lower bounds on the
 absolute indicator of any vectorial Boolean function and study the divisibility of the autocorrelation coefficients.
Besides, we study the invariance of the absolute indicator and of the autocorrelation spectrum under the affine, EA and CCZ equivalences. We also present all possible autocorrelation spectra  of optimal $4$-bit S-boxes. At the end of this section, we study some properties of the autocorrelation of APN, plateaued and AB functions. 
In Section \ref{polynomials}, we consider  the autocorrelation of some special polynomials. Finally, Section \ref{conclusion} draws some conclusions of our work. }
 
\section{Preliminaries}
\label{Preliminaries}

{In this section, we firstly recall some basics on (vectorial) Boolean functions and known results that are useful for our subsequent discussions. 
Since the vector space $\gf_2^n$ can be deemed as the finite field  $\gf_{2^n}$ for a fixed choice of basis, 
we will use the notation $\gf_2^n$ and $\gf_{2^n}$ interchangeably when there is no ambiguity.
We will also use the inner product $a\cdot b$ and $\tr_{2^n}(ab)$ in the context of vector spaces and finite fields  interchangeably. 
For any set $E$, we denote the nonzero elements of $E$ by $E^{*}$ (or $E\setminus\{0\}$) and the cardinality of $E$ by $\#E$.  
}

\subsection{Walsh transform, Bent functions, AB functions and Plateaued functions}
An $n$-variable Boolean function is a mapping from $\gf_{2}^n$ to $\gf_2$.
  For any $n$-variable Boolean function $f$, its \textit{Walsh transform} of $f$ is defined as 
$$ W_f(\omega) = \sum_{x\in\gf_{2}^n}(-1)^{f(x)+\omega\cdot x}, $$
where $``\cdot"$ is an inner product on $\gf_2^n$. 
{ The Walsh transform of $f$ can be seen as the \textit{discrete Fourier transform} of the function $(-1)^{f(x)}$ and yields the well-known Parseval's relation \cite{Carlet2010}} :
$$\sum_{\omega\in\gf_{2}^n}W_f^2(\omega)=2^{2n}.$$
The \textit{linearity} of $f$ is defined by $$\mathtt{L}(f) =  \max_{\omega\in\gf_{2}^n}|W_f(\omega)|$$ and \textit{nonlinearity} of $f$ is defined by
$$\mathtt{NL}(f) = 2^{n-1} - \frac{1}{2} \mathtt{L}(f), $$
where $| r|$ denotes the absolute value of any real value $r$.  
According to the Parseval's relation, it is easily seen that the nonlinearity of an $n$-variable Boolean function is upper bounded by $2^{n-1}-2^{n/2-1}$. 
 Boolean functions achieving the maximum nonlinearity are called \emph{bent} functions and exist only for even $n$; their Walsh transforms take only two values $ \pm 2^{n/2}$ \cite{Rothaus1976}.

For an $(n,m)$-function $F$ from $\gf_2^n$ to $\gf_2^m$, its \textit{component} corresponding to a nonzero $v\in \gf_2^m$ is the Boolean function given by 
$$ f_{v} (x) = v\cdot F(x).$$ For any $u\in\gf_{2}^n$ and nonzero $v\in\gf_{2}^m$, the Walsh transform of $F$ is defined by those of its components $f_v$, i.e.,
$$W_F(u,v) = \sum_{x\in\gf_{2}^n}(-1)^{u\cdot x + v\cdot F(x)}.$$ 
The linear approximation table (LAT) of an $(n,m)$-function $F$  is the $2^n\times 2^m$ table, in which the entry at position~$(u,v)$ is: $$\mathtt{LAT}_F(u,v) = W_F(u,v),$$ where  $u\in\gf_{2}^n$ and $v\in\gf_{2}^m$. 
The maximum absolute entry of the LAT, ignoring the $0$-th column, is the linearity of $F$ denoted as $\mathtt{L}(F)$, i.e., 
$$\mathtt{L}(F) = \max_{u\in\gf_{2}^n, v\in\gf_{2}^m\backslash\{0\}} |W_F(u,v)|$$ 
Similarly, the nonlinearity of $F$ is defined by the nonlinearities of the components, namely,
$$\mathtt{NL}(F) =  2^{n-1} - \frac{1}{2} \mathtt{L}(F).$$

An $(n,m)$-function $F$ is called \textit{vectorial bent}, or shortly \textit{bent} if 
all its components $F_v(x) = v\cdot F(x)$ for each nonzero $v\in \gf_2^m$ are bent. It is well-known $(n,m)$-bent functions exist only if $n$ is even and $m\le \frac{n}{2}$. 
Interested readers can refer to \cite{Sihem2016,Tokareva2015} for more results on bent functions.
For $(n,m)$-functions $F$ with $m\ge n-1$, the Sidelnikov-Chabaud-Vaudenay bound
$$\mathtt{NL}(F)\le 2^{n-1}-\frac{1}{2} \left( \frac{3\cdot 2^n - 2(2^n-1)(2^{n-1}-1)}{2^m-1} -2  \right)^{1/2}$$
gives a better upper bound for nonlinearity than the universal bound \cite{EC:ChaVau94}. 
When $n=m$ and $n$ is odd, the inequality becomes
$$\mathtt{NL}(F)\le 2^{n-1}-2^{\frac{n-1}{2}},$$
and it is achieved 
by the \emph{almost bent} (AB) functions.
 It is well-known that an $(n,n)$-function $F$ is AB if and only if its Walsh transform takes only three values $0, \pm 2^{\frac{n+1}{2}}$ \cite{EC:ChaVau94}.

A Boolean functions is called \emph{plateaued} if its Walsh transform takes at most three values: $0$ and $\pm \mu$ (where $\mu$, a positive integer, is called the \emph{amplitude} of the plateaued function). It is clear that bent and almost bent functions are plateaued. Because of Parseval's relation, the amplitude $\mu$ of any plateaued function must be of the form $2^r$ for certain integer $r\ge n/2$. An $(n,m)$-function is called \emph{plateaued} if all its components are plateaued, with possibly different amplitudes. In particular, an $(n,m)$-function $F$ is called \emph{plateaued with single amplitude} if all its components are plateaued with the same amplitude. 
It is clear that AB functions form a subclass of plateaued functions with the single amplitude $2^{\frac{n+1}{2}}$. 

\subsection{Differential uniformity and APN functions}
For an $(n,m)$-function $F$ and any $u\in\gf_2^n\backslash\{0\}$, the function 
$$D_uF(x) = F(x)+F(x+u)$$
is called the derivative of $F$ in direction $u$. 
The \textit{differential distribution table} (DDT) of $F$ is the $2^n\times 2^m$ table, in which the entry at position~$(u,v)$ is
$$ \mathtt{DDT}_F(u,v) =\# \{ x\in\gf_2^n ~|~ D_uF(x) = v \}, $$
where  $u\in\gf_{2}^n$ and $v\in\gf_{2}^m$. The \textit{differential uniformity}~\cite{EC:Nyberg93} of $F$ is defined as 
$$\delta_F=\max_{u\in\gf_{2}^n\backslash\{0\}, v\in\gf_{2}^m}\mathtt{DDT}_F(u,v).$$

Since $D_uF(x)=D_uF(x+u)$ for any $x, u$ in $\gf_2^n$, the entries of DDT are always even and the minimum of differential uniformity of $F$ is $2$. The functions with differential uniformity $2$ are called \emph{almost perfect nonlinear} (APN) functions.

{  
\subsection{The DLCT and the autocorrelation table} \label{Subsec-AAC}

Very recently, Bar-On et al. in \cite{EC:BDKW19} presented the concept of the differential-linear connectivity table (DLCT) of $(n,m)$-functions $F$. 

\begin{Def}
	\cite{EC:BDKW19}\label{DLCT}
	Let $F$ be an $(n,m)$-function. The DLCT of $F$ is the $2^n\times 2^m$ table whose rows correspond to input differences to $F$ and whose columns correspond to output masks of $F$, defined as follows: for $u\in\gf_2^n$ and $v\in\gf_2^m$, the DLCT entry at $(u, v)$ is defined by
	$$\mathtt{DLCT}_F(u, v)  = \#  \{ x \in\gf_2^n | v\cdot F(x) = v \cdot F(x+u)  \}   -2^{n-1}.  $$
\end{Def}

Since for any $u\in\gf_{2}^n\backslash\{0\}$, $D_uF(x) = D_uF(x+u)$, $\mathtt{DLCT}_F(u,v)$ must be even. Furthermore, for a given $u\in\gf_{2}^n\backslash\{0\}$, if  $D_uF$ is a $2\ell$-to-$1$ mapping for a positive integer $\ell$, then $\mathtt{DLCT}_F(u,v)$ is a multiple of $2\ell$. Moreover, it is trivial that for any $(u, v)\in\gf_2^n\times \gf_{2}^m$, $\left|  \mathtt{DLCT}_F(u, v)\right| \le 2^{n-1} $, and $\mathtt{DLCT}_F(u,v) = 2^{n-1}$ when either $u=0$ or $v=0$. Therefore, 
we only need to focus on the cases for $u\in\gf_{2}^n\backslash\{0\}$ and $v\in\gf_{2}^m\backslash\{0\}$.



Our first observation on the DLCT is that it coincides with the {\em autocorrelation table} (ACT) of $F$~\cite[Section~3]{ZZI2000}. 
Below we recall the definition of the autocorrelation of Boolean functions, see e.g.~\cite[P.~277]{Carlet2010}, and extend it to vectorial Boolean functions.

\begin{Def} \cite{ZZ1995} Given a Boolean function $f$ on $\gf_{2}^n$, the \textit{autocorrelation} of the function $f$ at $u$ is defined as 
	$$\AC_f(u) = \sum_{x\in\gf_{2}^n}(-1)^{f(x)+f(x+u)}.$$ 
	Furthermore, the
	\textit{absolute indicator} of $f$ is defined as $\Delta_f = \max_{u\in \gf_{2}^n \setminus \{0\}}|\AC_f(u)|$.
\end{Def}

Similarly to Walsh coefficients, this notion can naturally be generalized to vectorial Boolean functions as follows.
\begin{Def}
	\label{AC_Def}
	Let $F$ be an $(n,m)$-function. For any  $u\in\gf_{2}^n$ and $v\in\gf_{2}^m$,  the \textit{autocorrelation} of $F$ at $(u,v)$ is defined as
	$$\AC_F(u,v) = \sum_{x\in\gf_{2}^n}(-1)^{v\cdot (F(x) + F(x+u))},$$
	and the \textit{autocorrelation spectrum} of $F$ is given by the multiset
	$$ \Lambda_F = \Big\{ \AC_F(u,v): u\in\gf_{2}^n\backslash \{ 0 \}, v\in\gf_{2}^m \backslash \{ 0 \}   \Big\}. $$ 
	Moreover, the \textit{absolute indicator} of $F$ is defined as $$\Delta_F =\max_{u\in\gf_{2}^n\backslash\{0\}, v\in\gf_{2}^m\backslash\{0\}}|\AC_F(u,v)|.$$ 
\end{Def}
In \cite{ZZI2000}, the term Autocorrelation Table (ACT) for a vectorial Boolean function was introduced. Similarly to the LAT, it contains the autocorrelation spectra of the components of~\(F\):
\[\mathsf{ACT}_F(u,v) = \AC_F(u,v)\;.\]
It is also worth noticing that
\begin{equation}
	 	\AC_F(u,v) = W_{D_uF}(0, v).
\end{equation}

From Definitions \ref{DLCT} and \ref{AC_Def}, we immediately have the following connection between the DLCT and the autocorrelation of vectorial Boolean functions.		
		\begin{Prop}
			\label{DLCT_Delta}
			Let $F$ be an $(n,m)$-function. Then for any $u\in\gf_{2}^n$ and $v\in\gf_{2}^m$, the autocorrelation of $F$ at $(u,v)$ is twice the value of the DLCT of $F$ at the same position $(u,v)$, i.e., 
			\begin{equation*}
			\mathtt{DLCT}_F(u,v) = \frac{1}{2}\AC_F(u,v)\;.
			\end{equation*}
                        Moreover
                        \[\max_{u\in\gf_{2}^n\backslash\{0\}, v\in\gf_{2}^m\backslash\{0\}} |\mathtt{DLCT}_F(u,v)| = \frac{1}{2}\Delta_F\;.\]
	\end{Prop} 

\begin{proof}
	Denote $M_i = \{ x\in\gf_{2}^n |  v\cdot \left( F(x)  + F(x+u) \right) =i   \}$.
	From the definitions of DLCT and autocorrelation it follows that
				\begin{eqnarray*}
		2\cdot\mathtt{DLCT}_F(u, v) &=& 2 \cdot \# \{ x \in\gf_2^n | v\cdot F(x) = v \cdot F(x+u)  \}   - 2^{n} \\
		&	=& \# M_0 -( 2^n - \#M_0)  \\ 
		& =& \#M_0-\#M_1	
\\		&= & \sum_{x\in\gf_{2}^n}(-1)^{v\cdot (F(x) + F(x+u))} 
\\&=& \AC_F(u,v).
	\end{eqnarray*}	
This gives the desired conclusion.
\end{proof} 	
For the remainder of this paper we thus stick to the established notion of the autocorrelation table instead of the DLCT, and we will study the absolute indicator of the function since it determines the highest magnitude in the DLCT.	 
\begin{Rem}
	Let us recall some relevant results on the autocorrelation table.	The entries $\AC_{F}(u,v), v \neq 0$ in each nonzero 	
	row in the ACT of an $(n,n)$-function $F$ sum to zero if and only if $F$ is a permutation (see e.g. \cite[Proposition 2]{BCC2006}).  The same property holds when the entries $\AC_{F}(u,v)$,  $u\neq0$  in each nonzero column in the ACT are considered (see e.g. \cite[Eq. (9)]{BCC2006}).
\end{Rem}
}

\section{Some characterizations and properties of the autocorrelation table}
\label{Characterizations}

{In this section, we give some characterizations and properties of the DLCT of vectorial Boolean functions from the viewpoint of the autocorrelation introduced in Subsection \ref{Subsec-AAC}. }

\subsection{Links between the autocorrelation and the Walsh transform}
In this subsection, we express the autocorrelation (or equivalently the DLCT) by the Walsh transform of the function. The following proposition shows that the restriction of the autocorrelation function $u \mapsto \AC_{F} (u,v) $ can be seen as the discrete Fourier transform of the squared Walsh transform of \(F_v\): $\omega\mapsto W_F(\omega,v)^2$.

\begin{Prop}
	\label{DLCT_WT}
	Let $F$ be an $(n,m)$-function.
        Then for any $u\in\gf_{2}^n$ and $v\in\gf_{2}^m$, 
 \[W_F(u,v)^2 = \sum_{\omega \in\gf_2^n}(-1)^{\omega\cdot u}\AC_F(u,v).\]
         Conversely, the inverse Fourier transform leads to
  	 \begin{equation}
  \label{DW}
  \AC_F(u,v) = \frac{1}{2^{n}}\sum_{\omega\in\gf_2^n}(-1)^{u\cdot \omega} W_F(\omega,v)^2
  \end{equation} 
Moreover, we have 
$$\sum_{u\in\gf_{2}^n} 	\AC_F(u,v) = W_F(0,v)^2 $$
and
\begin{equation}
\label{D2W4}
\sum_{u\in\gf_{2}^n}\AC_F(u,v)^2 = \frac{1}{2^{n}}\sum_{\omega\in\gf_{2}^n}W_F(\omega,v)^4. 
\end{equation} 
	 \end{Prop}

\begin{proof}
According to the definition, for any $u\in\gf_{2}^n$,
	\begin{eqnarray*}
		W_F(u,v)^2 &=& \sum_{x\in\gf_{2}^n}(-1)^{u\cdot x + v\cdot F(x)}\sum_{y\in\gf_{2}^n}(-1)^{u\cdot y + v\cdot F(y)} \\
		&=&\sum_{x,y\in\gf_2^n}(-1)^{u\cdot(x+y)+ v\cdot (F(x) + F(y))} \\
		&=&\sum_{x,\omega\in\gf_2^n}(-1)^{u\cdot \omega + v\cdot (F(x) + F(x+\omega))} \\
		&=&\sum_{\omega\in\gf_2^n}(-1)^{u\cdot \omega}\sum_{x\in\gf_2^n}(-1)^{v\cdot (F(x) + F(x+\omega))}\\
		&=&\sum_{\omega\in\gf_2^n}(-1)^{u\cdot \omega}\AC_F(\omega,v).
	\end{eqnarray*}
	
	The inverse Fourier Transform then leads to
	$$\AC_F(u,v) = \frac{1}{2^n}\sum_{\omega\in\gf_2^n}(-1)^{\omega\cdot u} W_F(\omega,v)^2. $$
	Moreover, we have 
	\begin{eqnarray*}
	\sum_{u\in\gf_2^n}\AC_F(u,v) 
		&=& \frac{1}{2^{n}} \sum_{\omega\in\gf_{2}^n} W_F(\omega,v)^2\sum_{u\in\gf_2^n}(-1)^{\omega\cdot u}\\
		&=& W_F(0,v)^2.
	\end{eqnarray*}
	Furthermore, Parseval’s equality leads to
        \[\sum_{u\in\gf_{2}^n}\AC_F(u,v)^2 = \frac{1}{2^{n}}\sum_{\omega\in\gf_{2}^n}W_F(\omega,v)^4. \]
\end{proof}

{
\begin{Rem}
		\emph{ It should be noted that the relations Eq. (\ref{DW}) and Eq. (\ref{D2W4}) were already obtained in \cite{SAC:GonKho03} and \cite{ZZ1995} for Boolean functions.
		Here we generalize the results to vectorial Boolean functions. }
\end{Rem} 
}

\subsection{Links between the autocorrelation and the DDT}
Zhang et al. in \cite[Section~3]{ZZI2000} showed that, for an \((n,n)\)-function, the row of index~\(a\) in the autocorrelation table 
$b \mapsto \AC_{F}(a,b)$ corresponds to the Fourier transform of the row of index~$a$ in the DDT: $v \mapsto \mathtt{DDT}_F(a,v)$. 
This relation coincides with the one provided in \cite[Proposition 1]{EC:BDKW19}. We here express it in the case of \((n,m)\)-functions.
It is worth noticing that this correspondence points out the well-known relation between the Walsh transform of~$F$ and its DDT exhibited by~\cite{EC:ChaVau94,EC:BloNyb13}.

\begin{Prop}
	\label{DLCT_DDT}
	Let $F$ be an $(n,m)$-function. Then, for any $u\in\gf_{2}^n$ and $v\in\gf_{2}^m$, we have
	
		\begin{eqnarray*}
		\label{DLCT_1_DDT}
		\AC_F(u,v) & = & \sum_{\omega\in\gf_{2}^m}(-1)^{v\cdot\omega} \mathtt{DDT}_F(u,\omega)\\
                \mathtt{DDT}_{F}(u,v) &= &2^{-m} \sum_{\omega \in \gf_2^m} (-1)^{v \cdot \omega} \AC_{F}(u,\omega).
		\end{eqnarray*}
                Most notably,
		$$\sum_{v\in\gf_{2}^m} \AC_F(u,v) = 2^{m} \mathtt{DDT}_F(u,0)  $$ 
		implying  $$\sum_{u\in\gf_2^n,v\in\gf_2^m}\AC_F(u,v)= 2^{m+n},$$
	        and
		\begin{equation}
		\label{D2DDT2}
		\sum_{v\in\gf_2^m} \AC_F(u,v)^2 = 2^{m} \sum_{\omega\in\gf_2^m} \mathtt{DDT}_F(u,\omega)^2.
		\end{equation}
	
\end{Prop}

\begin{proof}
	
		The first equation holds since  
		\begin{eqnarray*}
		\AC_F(u,v) & = &  \sum_{x\in\gf_{2}^n} (-1)^{v\cdot (F(x)+F(x+u))} \\
		&=& \sum_{\omega\in\gf_{2}^m}(-1)^{v\cdot\omega} \mathtt{DDT}_F(u,\omega)\;.
		\end{eqnarray*}
                The inverse Fourier transform then leads to
                \[\mathtt{DDT}_{F}(u,v) =2^{-m} \sum_{\omega \in \gf_2^m} (-1)^{v \cdot \omega} \AC_{F}(u,\omega).\]
                By applying this relation to \(v=0\), we get
		\[2^m \mathtt{DDT}_F(u,0) = \sum_{\omega \in \gf_2^m}\AC_{F}(u,\omega).\]
                Obviously, we deduce that $\sum_{u\in\gf_2^n,v\in\gf_2^m}\AC_F(u,v)= 2^{m+n}$ holds.
                Moreover, Parseval's relation implies
		\[\sum_{v\in\gf_2^m} \AC_F(u,v)^2 = 2^{m} \sum_{\omega\in\gf_2^m} \mathtt{DDT}_F(u,\omega)^2.\]
\end{proof}

\subsection{Bounds on the absolute indicator}

Similar to other cryptographic criteria, it is interesting and important to know 
how ``good" the absolute indicator of a vectorial Boolean function could be.
It is clear that the absolute indicator of any $(n,m)$-function is upper bounded by $2^{n}$. But finding its smallest possible value is an open question investigated by many authors.
From the definition, the autocorrelation spectrum of \(F\) equals $\{0\}$ if and only if $F$ is a bent function, which implies that $n$ is even and $m\le \frac{n}{2}$.
However, finding lower bounds in other cases is much more difficult.
For instance, Zhang and Zheng conjectured~\cite[Conjecture 1]{ZZ1995} that the absolute indicator of a balanced Boolean function of $n$~variables was at least $2^{\frac{n+1}{2}}$.
But this was later disproved first for odd values of~$n \geq 9$ by modifying the Patterson-Wiedemann construction, namely for $n \in \{9, 11\}$ in~\cite{TIT:KavMaiYuc07}, for $n=15$ in~\cite{TIT:MaiSar02,DAM:Kavut16} and for $n=21$ in~\cite{DM:GanKesMai06}.
For the case $n$ even, \cite{TIT:TanMai18} gave a construction for balanced Boolean functions with absolute indicator strictly less than $2^{n/2}$ when $n \equiv 2 \bmod{4}$.
Very recently, similar examples for $n \equiv 0 \bmod{4}$ were exhibited by~\cite{DCC:KavMaiTan19}.
However, we now show that such small values for the absolute indicator cannot be achieved for \((n,n)\)-\emph{vectorial functions}.

Proposition \ref{DLCT_DDT} leads to the following upper bound on the sum of all squared autocorrelation coefficients in each row. This result can be found in~\cite{FSE:Nyberg94} (see also \cite[Theorem~2]{BCC2006}) in the case of \((n,n)\)-functions. We here detail the proof in the case of \((n,m)\)-functions for the sake of completeness.
\begin{Prop}\label{prop:boundAPN}
    Let $F$ be an \((n,m)\)-function.
    Then, for all $u \in \gf_2^n$, we have
    \begin{equation*}
        \sum_{v \in \gf_2^m} \AC_F(u,v)^2 \geq 2^{n+m+1}\;.
    \end{equation*}
    Moreover, equality holds for all nonzero~$u \in \gf_2^n$ if and only if $F$ is APN.
\end{Prop}
\begin{proof}
  From~(\ref{D2DDT2}), we have that, for all $u \in \gf_2^n$,
    \begin{equation*}
        \sum_{v \in \gf_2^m} \AC_{F}(u,v)^2 = 2^m \sum_{\omega\in\gf_2^m} \mathtt{DDT}_F(u,\omega)^2
    \end{equation*}
    Cauchy-Schwarz inequality implies that
    \begin{equation*}
      \left(\sum_{\omega \in \gf_2^m} \mathtt{DDT}_{F}(u,\omega)\right)^2 \leq \left(\sum_{\omega \in \gf_2^m} \mathtt{DDT}_{F}(u,\omega)^2\right) \times \#\{\omega\in \gf_2^m | \mathtt{DDT}_{F}(u,\omega) \neq 0\}\;,
    \end{equation*}
    with equality if and only if all nonzero elements in $\{\mathtt{DDT}_{F}(u,\omega) | \omega \in \gf_2^m\}$ are equal.
    Using that
    \begin{equation*}
        \#\{\omega \in \gf_2^m | \mathtt{DDT}_{F}(u,\omega) \neq 0\} \leq {2^{n-1}}
    \end{equation*}
    with equality for all nonzero $u$ if and only if $F$ is APN, we deduce that
    \begin{equation*}
        \sum_{\omega\in \gf_2^m} \mathtt{DDT}_{F}^2(u,\omega) \geq 2^{n+1}
    \end{equation*}
    with equality for all nonzero $u$ if and only if $F$ is APN.
    Equivalently, we deduce that
    \begin{equation*}
        \sum_{v \in \gf_2^m} \AC_{F}^2(u,v) \geq 2^{n+m+1}
    \end{equation*}
    with equality for all nonzero $u$ if and only if $F$ is APN.
    \end{proof}

From the lower bound on the sum of all squared entries within a row of the autocorrelation table, we deduce the following lower bound on the absolute indicator.
\begin{Th} \label{Th_Bound}
	Let $F$ be an $(n,m)$-function, where $m\ge n$. Then 
	\begin{equation}
	\label{DLU_bound}
	\Delta_F\ge \sqrt{\frac{2^{m+n+1}-2^{2n}}{2^m-1}}. 
	\end{equation}
        Most notably, if \(m=n\),
        \[\Delta_F > 2^{n/2}\;.\]
\end{Th}
\begin{proof}
  From the previous proposition, we deduce that
  \[\sum_{v \in \gf_2^m \backslash\{0\}} \AC_{F}(u,v)^2 \geq 2^{n+m+1}- 2^{2n}\;.\]
  Since
  \begin{equation*}
        \sum_{v \in \gf_2^m \backslash\{0\}} \AC_{F}(u,v)^2 \leq \Delta_{F}^2 (2^m-1)\;,
    \end{equation*}
  the result directly follows.
  When \(m=n\), the bound corresponds to
  \[\sqrt{\frac{2^{2n}}{2^n-1}} > 2^{n/2}\;.\]
\end{proof}

	Note that the condition $m \geq n$ in Theorem~\ref{Th_Bound} is to ensure
	 the term under the square root is strictly greater than~\(0\).

%


\subsection{Divisibility of the autocorrelation} 

In this subsection, we investigate the divisibility property of the autocorrelation coefficients of vectorial Boolean functions.   
	\begin{Prop}\label{prop-divisibility}
		Let $n>2$ and $F: \gf_{2}^n \to \gf_{2}^m$ be a vectorial Boolean function with algebraic degree at most~$d$. Then, for any $u\in\gf_{2}^n$ and $v\in\gf_2^m$, $\AC_F(u,v)$ is divisible by $2^{\lceil \frac{n-1}{d-1} \rceil+1 }.$  In particular, when $m=n$ and $F$ is a permutation,   $\AC_F(u,v)$ is divisible by $8$. 
	\end{Prop}
	\begin{proof}
		By definition,  for any $u\in\gf_{2}^n$ and $v\in\gf_2^m$,
		$$\AC_F(u,v) = W_{D_uf_v} (0).$$
		Note that for given $u\in\gf_{2}^n$ and $v\in\gf_2^m$, the Boolean function $$h_{u,v} = D_uf_v=v\cdot(F(x)+F(x+u)),$$ satisfies two properties: $\deg(h_{u,v})\le d-1$ since $F$ has degree at most $d$ and $h_{u,v}(x) = h_{u,v}(x+u).$ 
		
		We now focus on the divisibility of $W_{h_{u,v}}(0)$. First, assume for simplicity that $u = e_n=(0, \cdots, 0, 1)$, we discuss the general case afterwards. 
		Since $h_{e_n,v}(x+e_n) = h_{e_n,v}(x)$, the value of $h_{e_n,v}(x)$ is actually determined by the first $(n-1)$ coordinates of $x$. Hence $h_{e_n,v}(x)$ can be expressed as $h_{e_n,v}(x)=h(x')\,: \gf_{2}^{n-1} \rightarrow \gf_2$ and  the Walsh transform of $h_{e_n,v}$ at point $0$ satisfies
		$$ W_{h_{e_n,v}}(0) = \sum_{x'\in\gf_{2}^{n-1}, x_n\in\gf_{2}} (-1)^{h_{e_n,v}(x',x_n)} = 2 \cdot \sum_{x'\in\gf_{2}^{n-1}} (-1)^{h(x')} = 2\cdot W_h(0).  $$
		 It is well-known that the values taken by the Walsh transform of a Boolean function $f$ from $\gf_{2}^n$ to $\gf_{2}$ with degree $d$ are divisible by $2^{\lceil \frac{n}{d-1} \rceil}$ (see \cite{McEliece1972} or \cite[Section 3.1]{Carlet2010}). We then deduce that $W_h(0)$ is divisible by $2^{\lceil \frac{n-1}{d-1} \rceil}$, implying that $W_{h_{e_n,v}}(0)$ is divisible by $2^{\lceil \frac{n-1}{d-1} \rceil+1}$. Most notably, if $m=n$ and $F$ is bijective, then $d<n$. We then have that $$\Big\lceil \frac{n-1}{d-1}\Big\rceil\ge 2,$$ implying that $\AC_F(u,v)$ is divisible by $8$.
		
		In the case that $u\neq e_n$, we can find a linear transformation $L$ such that $L(e_n)=u,$ with which we have the affine equivalent function $G=F\circ L.$ Due to the affine invariance of $G$'s and $F$'s autocorrelation spectra, the same holds for $\AC_G(u,v)$ in this case. 
	\end{proof}
	
	In particular, for $(n, m)$-functions of algebraic degree $3$, we have the following result.
	\begin{Prop} \label{prop-cubic}
		Suppose an $(n, m)$-functions  $F$ has algebraic degree~$3$.
		Then for nonzero $u$ and $v$, we have
		\begin{equation*}
		|\AC_{F}(u,v)| \in \left\{0, 2^{\frac{n+d(u,v)}{2}}\right\},
		\end{equation*}
		where $d(u,v)  = \dim \left\{w\in \gf_{2}^n ~|~ D_{u}D_{w}f_{v} = c\right\}$ and $c \in \gf_2$ is constant.
	\end{Prop}
	
	\begin{proof} Since $F$ has algebraic degree $3$, the derivative of order~two $D_uD_wf_{v}(x)=A_{u,v}(w)\cdot x+C_{u,v}(w)$ is affine over $\gf_{2^n}$, where 
		$A_{u,v}(w)$ and $C_{u,v}(w)$ belong to $\gf_2$. Moreover, 
		the function $w \mapsto C_{u,v}(w)$ is linear over the linear subspace $L(u, v) = \{ w \in \gf_{2}^n : A_{u,v}(w) = 0\}= \{ w \in \gf_{2}^n : D_uD_wf_{v}(x) = C_{u,v}(w)\}$.
		From the definition of autocorrelation,  we have 
		\begin{equation}
		\begin{array}{rcl}
		\AC_F(u, v)^2 & = & \left(\sum\limits_{x\in \gf_{2^n}}(-1)^{v\cdot (F(x+u)+F(x)}\right)^2 
		\\ &= & \sum\limits_{x, y\in \gf_{2}^n}(-1)^{v\cdot(F(x+u)+F(x) + F(y+u) +F(y))}
		\\&=& \sum\limits_{x, w\in \gf_{2}^n}(-1)^{v\cdot(F(x+u)+F(x) + F(x+w+u) +F(x+w))}
		\\&=& \sum\limits_{x, w\in \gf_{2}^n}(-1)^{D_u D_wf_v(x)}
		\\&=& \sum\limits_{w\in \gf_{2}^n}(-1)^{C_{u,v}(w)} \sum\limits_{x\in \gf_{2^n}}(-1)^{A_{u,v}(w)\cdot x}.
		\end{array}
		\end{equation}
		Hence, 
		$$
		\AC_F(u, v)^2 = \begin{cases}
		0, & \text{ if } A_{u,v}(w) \neq 0,  \\
		2^{n+d(u,v)} & \text{ if } A_{u,v}(w) = 0 \text{ and } C_{u,v}(w) = c \text{ in } L(u,v).
		\end{cases} 
		$$ 
		The desired conclusion directly follows.
	\end{proof}
    Proposition \ref{prop-cubic} implies that any entry in the autocorrelation table of a cubic function is 
    divisible by $2^{\frac{n+d}{2}}$, where $d$ is the smallest integer among $d(u,v)$ when $u$, $v$ run through 
    $\gf_{2}^n\setminus \{0\}$ and       $\gf_{2}^m\setminus \{0\}$, respectively. 
    It is clear that $d\geq 1$. Furthermore, when $d\geq 2$, Proposition 
    \ref{prop-cubic}
    improves the result in Proposition \ref{prop-divisibility}.

\subsection{Invariance under Equivalence Relations}

Let $n, m$ be two positive integers. There are several equivalence relations of functions from $\gf_2^n$ to $\gf_2^m$  and they play vital roles in classifying functions with good properties, like AB and APN functions \cite{BCP2006}. In this subsection, we first recall three equivalence relations, i.e., affine, EA and CCZ~\cite{CCZ98}. Then we study the autocorrelation and related concepts with respect to these equivalence relations.

\begin{Def}
	\cite{Bud2014}
	Let $n,m$ be two positive integers. Two functions $F$ and $F^{'}$ from $\gf_2^n$ to $\gf_2^m$ are called 
	\begin{enumerate}
		\item {affine equivalent (resp. linear equivalent) }  if $F^{'} = A_1\circ F\circ A_2$, where the mappings $A_1$ and $A_2$ are affine (resp. linear) permutations of $\gf_2^m$ and $\gf_2^n$, respectively;
		\item { extended affine equivalent} (EA equivalent) if $F^{'} =  A_1\circ F \circ A_2 + A$, where the mappings $A : \gf_2^n\to\gf_2^m, A_1 : \gf_2^m\to\gf_2^m,  A_2 : \gf_2^n\to\gf_2^n$ are affine and where $A_1$ and $A_2$ are permutations;
		\item { Carlet-Charpin-Zinoviev equivalent } (CCZ equivalent) if for some affine permutation $\mathcal{L}$ over $\gf_2^n \times \gf_2^m$, the image by \(\mathcal{L}\) of the graph of $F$ is the graph of $F^{'}$, that is $\mathcal{L}(G_F)=G_{F^{'}}$, where $G_F=\{ (x,F(x)) | x\in\gf_2^n  \}$ and $G_{F^{'}}= \{ (x,F^{'}(x)) | x\in\gf_2^n  \} $.
	\end{enumerate}
\end{Def}

It is known that affine equivalence is a particular case of EA-equivalence, which is again a particular case of CCZ-equivalence. In addition, every permutation is CCZ-equivalent to its compositional inverse.  
Two important properties of cryptographic functions, the differential uniformity and the nonlinearity, are invariant under CCZ-equivalence.  
However, as we will show in this subsection, 
the autocorrelation spectrum is invariant under affine equivalence,
and further its extended autocorrelation spectrum, i.e., the multiset
$\{ |\AC_F(u,v)|\; :\; u \in \gf_{2}^n, v\in \gf_{2}^m \}$, is
invariant under extended affine equivalence. However, they are generally not invariant under compositional inverse, thereby are not invariant under CCZ-equivalence.

\begin{Th}
	Assume two $(n, m)$-functions $F$ and $F^{'}$ are EA-equivalent, then the extended autocorrelation spectrum of $F$ equals that of $F'$.
	In particular, if they are affine equivalent, then the autocorrelation spectrum of $F$ equals that of $F'$.
\end{Th}

\begin{proof}
	Since $F$ and $F^{'}$ are EA equivalent, there exist affine mappings  $A : \gf_2^n\to\gf_2^m, A_1 : \gf_2^m\to\gf_2^m,  A_2 : \gf_2^n\to\gf_2^n$, where $A_1, A_2$ are permutations, such that $F^{'} =  A_1\circ F \circ A_2 + A$. Assume that the linear parts of $A, A_1, A_2$ are $L, L_1, L_2$ respectively. Then for any  $u\in\gf_{2}^n\backslash\{0\}$ and $v\in\gf_{2}^m\backslash\{0\}$,
	\begin{eqnarray*}
		\AC_{F^{'}}(u,v) & = & \sum_{x\in\gf_2^n}(-1)^{v\cdot \left( F^{'}(x) + F^{'}(x+u) \right)} \\
		&=& \sum_{x\in\gf_2^n} (-1)^{ v\cdot  \left( A_1\circ F \circ A_2(x) + A(x) +   A_1\circ F \circ A_2(x+u) + A(x+u)     \right)  } \\
		&=& (-1)^{v\cdot L(u)} \sum_{x\in\gf_2^n} (-1)^{v\cdot  \left( A_1\circ F \circ A_2(x) +   A_1\circ F \circ A_2(x+u)   \right)} \\
		&=&  (-1)^{v\cdot L(u)} \sum_{x\in\gf_2^n} (-1)^{ v \cdot L_1 \left(  F \circ A_2(x) + F \circ A_2(x+u) \right)   } \\
		&=& (-1)^{v\cdot L(u)} \sum_{x\in\gf_2^n} (-1)^{L_1^{\mathrm{T}}(v) \cdot \left(  F \circ A_2(x) + F \circ A_2(x+u) \right)  }\\
		&=& (-1)^{v\cdot L(u)}  \sum_{y\in\gf_2^n} (-1)^{L_1^{\mathrm{T}}(v) \cdot \left( F(y) + F  \left( y+L_2(u)  \right)   \right) } \\
		&=& (-1)^{v\cdot L(u)} 	\AC_{F}(L_2(u),L_1^{\mathrm{T}}(v) ),
	\end{eqnarray*}
	where $L_1^{\mathrm{T}}$ denotes the transpose of $L_1$.  Moreover, when $F$ and $F^{'}$ from $\gf_{2}^n$ to $\gf_{2}^m$ are affine equivalent, namely, $A=0$, we have 
	$$ 	\AC_{F^{'}}(u,v) =  \AC_{F}(L_2(u),L_1^{\mathrm{T}}(v) ). $$
\end{proof}

To examine the behavior under CCZ equivalence, we focus on the autocorrelation of a permutation and the autocorrelation of its compositional inverse.
When $n=m$ and $F$ permutes $\gf_{2}^n$, Zhang et al. showed in \cite[Corollary 1]{ZZI2000} that $$ \mathtt{ACT}_{F^{-1}} = H^{-1}\cdot \mathtt{ACT}_F \cdot H, $$ which in our notation is 
\begin{equation}\label{Eq-DLCT-Inv}
\AC_{F^{-1}}(u,v) = \frac{1}{2^n}\sum_{a,b\in\gf_{2}^n}(-1)^{u\cdot a+ v\cdot b} \AC_{F}(a,b). 
\end{equation}
The relation in Eq. \eqref{Eq-DLCT-Inv} indicates that the autocorrelation spectrum of an $(n,n)$-permutation
$F$ is in general not equal to that of $F^{-1}$.

This observation is indeed confirmed by many examples, in which an $(n,n)$-permutation $F$ has linear structures but its inverse has not.  Recall from \cite{ZZI2000}  that a linear structure for an $(n,m)$-function $F$ is a tuple $(u, v)\in \gf_{2^n}\times \gf_{2^m}$ such that	
$x \mapsto v\cdot (F(x)+F(x+u))$ is constant, zero or one, and $\AC_F(u,v)=\pm 2^{n}$ if and only $(u, v)$ forms a linear structure.
For instance, the S-boxes from \textsc{safer}~\cite{FSE:Massey93}, \textsc{SC2000}~\cite{FSE:SYYTIYTT01}, and \textsc{Fides}~\cite{CHES:BBKMW13} have linear structures in one direction but not in the other direction. This is also the case of the infinite family formed by the Gold permutations as analyzed in Section~\ref{sec:quadratic}.

Below, we also provide an example that demonstrates that the autocorrelation spectrum is not invariant under EA-equivalence.

\begin{example}
  Let $F(x) = \frac{1}{x}\in\gf_{2^7}[x]$ and $F^{'}(x) = \frac{1}{x}+x$. Then $F$ and $F^{'}$ are EA-equivalent. However, $\Lambda_F = \{ -24, -16,-8, 0, 8, 16\}$ while $\Lambda_{F^{'}} = \{ -24, -16,-8, 0, 8, 16, 24\}$. 
\end{example}

In \cite{LP2007}, the authors classified all optimal permutations over $\gf_{2}^4$ having the best differential uniformity and nonlinearity (both $4$) up to affine equivalence and found that there are only $16$ different optimal S-boxes, see Table \ref{Tab-4Sbox}. Based on the classification of optimal S-boxes, we exhaust all possibilities of the autocorrelation spectra of optimal S-boxes in Table \ref{Tab-4Autocorrelation}, where the subscript of each autocorrelation value indicates the number of its occurrences in the spectrum. 

\begin{table}[!htbp]
	\caption{Representatives for all $16$ classes of optimal $4$ bit Sboxes }\label{Tab-4Sbox}
	\centering
	\begin{tabular}{cc}	
		\hline
		$F_0$ ~&~  $ 0, 1, 2, 13, 4, 7, 15, 6, 8, 11, 12, 9, 3, 14, 10, 5 $ \\
		$F_1$ ~&~	$0, 1, 2, 13, 4, 7, 15, 6, 8, 11, 14, 3, 5, 9, 10, 12 $ \\
		$F_2$	~&~ $0, 1, 2, 13, 4, 7, 15, 6, 8, 11, 14, 3, 10, 12, 5, 9 $ \\
		$F_3$	~& ~   $0, 1, 2, 13, 4, 7, 15, 6, 8, 12, 5, 3, 10, 14, 11, 9$    \\
		$F_4$	 ~&~	$ 0, 1, 2, 13, 4, 7, 15, 6, 8, 12, 9, 11, 10, 14, 5, 3 $ \\
		$F_5$	 ~&~  $0, 1, 2, 13, 4, 7, 15, 6, 8, 12, 11, 9, 10, 14, 3, 5 $  \\
		$F_6$ ~&~  $0, 1, 2, 13, 4, 7, 15, 6, 8, 12, 11, 9, 10, 14, 5, 3$       \\
		$F_7$ ~&~	$0, 1, 2, 13, 4, 7, 15, 6, 8, 12, 14, 11, 10, 9, 3, 5$ \\
		$F_8$	~&~ $0, 1, 2, 13, 4, 7, 15, 6, 8, 14, 9, 5, 10, 11, 3, 12$ \\
		$F_9$	~& ~  $ 0, 1, 2, 13, 4, 7, 15, 6, 8, 14, 11, 3, 5, 9, 10, 12 $ \\
		$F_{10}$	 ~&~ $0, 1, 2, 13, 4, 7, 15, 6, 8, 14, 11, 5, 10, 9, 3, 12$	 \\
		$F_{11}$	 ~&~ $0, 1, 2, 13, 4, 7, 15, 6, 8, 14, 11, 10, 5, 9, 12, 3$  \\
		$F_{12}$	~&~ $0, 1, 2, 13, 4, 7, 15, 6, 8, 14, 11, 10, 9, 3, 12, 5$ \\
		$F_{13}$	~& ~  $0, 1, 2, 13, 4, 7, 15, 6, 8, 14, 12, 9, 5, 11, 10, 3$\\
		$F_{14}$	 ~&~ $0, 1, 2, 13, 4, 7, 15, 6, 8, 14, 12, 11, 3, 9, 5, 10$	 \\
		$F_{15}$	 ~&~ $0, 1, 2, 13, 4, 7, 15, 6, 8, 14, 12, 11, 9, 3, 10, 5$ \\	
		\hline
	\end{tabular}	
\end{table}

\begin{table}[!htbp] \renewcommand{\arraystretch}{1.3}
	\caption{Autocorrelation spectrum of $F_i$ for $ 0\le i\le 15$}\label{Tab-4Autocorrelation}
	\centering
	\begin{tabular}{c | c}	
		\hline
		$F_i$ &  Autocorrelation spectrum \\
		\hline
		$i \in \{ 3, 4, 5, 6, 7,11, 12, 13\}$ & $\left\{ -8^{60}, 0^{135}, 8^{30}  \right\}$ \\
		$i \in \{ 0,1,2,8\}$ & $\left\{ -16^{6}, -8^{48}, 0^{144}, 8^{24}, 16^3  \right\}$ \\
		$i \in \{ 9,10,14,15\}$ & $\left\{ -16^{2}, -8^{56}, 0^{138}, 8^{28}, 16^1  \right\}$ \\
		\hline
	\end{tabular}
\end{table}

\subsection{{Autocorrelation of Plateaued, AB and APN functions}}
APN and AB functions provide optimal resistance against differential attacks and linear attacks, respectively. Many researchers have studied some other properties of APN and AB functions (see for example \cite{Bud2014}). 
This subsection will investigate the autocorrelation of these optimal functions. We start with a general result for plateaued functions, which generalizes a result from~\cite{SAC:GonKho03}, where the authors studied the autocorrelation of 
		a plateaued Boolean function $f$ in terms of its dual function.

\begin{Prop}
	\label{Prop_AB}
	Let \(F\) be an $(n,m)$-plateaued function. For $v\in\gf_2^m\backslash\{0\}$, we denote the amplitude of the component~\(F_v\) by $2^{r_v}$ and define a  dual
	Boolean function of $f_v$ as
	\begin{equation}
\label{gv}
\widetilde{f}_v(b) = 	\left\{
\begin{array}{lr}
1,  &~  \text{if}~ W_{f_v}(b) \neq0, \\
0, &~ \text{if}~ W_{f_v}(b)=0.
\end{array}
\right.
\end{equation}	
	Then 
	$$\AC_F(u,v) = - 2^{2r_v-n-1}W_{\widetilde{f}_v}(u).$$
	Furthermore, when $F$ is an AB function from $\gf_2^n$ to itself, namely, $r_v=2^{\frac{n+1}{2}}$ for any $v\in\gf_{2}^n\backslash\{0\}$, $$\AC_F(u,v) = -W_{\widetilde{f}_v}(u). $$
\end{Prop}

\begin{proof}
	According to Eq. (\ref{DW}), we have 
	\begin{eqnarray*}
		\AC_F(u,v) &=& \frac{1}{2^{n}}\sum_{\omega\in\gf_2^n}(-1)^{u\cdot \omega} W_F(\omega,v)^2 \\
		&=& 2^{2r_v-n}\sum_{\omega\in\gf_2^n}(-1)^{u\cdot \omega} \widetilde{f}_v(\omega)\\
		&=& 2^{2r_v-n}\sum_{\omega\in\gf_2^n}\left( \frac{1}{2}\left(1-(-1)^{\widetilde{f}_v(\omega)}\right)  \right)(-1)^{u\cdot \omega}\\
		&=& -2^{2r_v-n-1}\sum_{\omega\in\gf_2^n}(-1)^{\widetilde{f}_v(\omega)+u\cdot \omega}\\
		&=& -2^{2r_v-n-1}W_{\widetilde{f}_v}(u).
	\end{eqnarray*}
	Particularly, when $F$ is an AB function, i.e., $r_v=\frac{n+1}{2}$ for any $v\in\gf_2^m\backslash\{0\}$, it is clear that $\AC_F(u,v) = -W_{\widetilde{f}_v}(u).$
\end{proof}

Similar to the AB functions, the autocorrelation of APN functions can also be expressed in terms of the Walsh transforms of some balanced Boolean functions.

\begin{Prop}
	\label{Prop_APN}
	Let $F$ be an APN function from $\gf_{2}^n$ to itself. For any nonzero $u\in \gf_{2}^n$, we define the Boolean function
	
	\begin{equation}
	\label{fu}
	\gamma_u(x) = 	\left\{
	\begin{array}{ll}
	1,  &~  \text{if}~ x \in \mathsf{Im}(D_uF), \\
	0, &~ \text{if}~ x\in\gf_{2}^n \backslash \mathsf{Im}(D_uF).
	\end{array}
	\right.
	\end{equation} Then the autocorrelation of $F$ can be expressed by the Walsh transform of $\gamma_u$ as
	$$\AC_F(u,v) = -W_{\gamma_u}(v).$$
\end{Prop}
\begin{proof} Since the APN function $F$ has a $2$-to-$1$ derivative function $D_uF(x)$ at any nonzero $u$,  we know that $\mathsf{Im}(D_uF)$ has cardinality $2^{n-1}$. Then,
	\begin{eqnarray*}
		\AC_F(u,v) &=& \sum_{x\in\gf_{2}^n}(-1)^{v\cdot (F(x+u)+F(x))} \\
		&=& 2\sum_{y\in \mathsf{Im}(D_uF)} (-1)^{v \cdot y}\\
		&=& \sum_{y\in \mathsf{Im}(D_uF)}(-1)^{v \cdot y} - \sum_{y\in \gf_{2}^n \backslash \mathsf{Im}(D_uF)}(-1)^{v \cdot y}\\
		&=& -\sum_{y\in\gf_{2}^n}(-1)^{\gamma_u(y)+v \cdot y}\\
		&=& -W_{\gamma_u}(v).
	\end{eqnarray*}
	\end{proof}
\mkq{
%
%
%
%
}

{From Proposition~\ref{Prop_APN}, we see that the autocorrelation of any APN function corresponds to the Walsh transform of the Boolean function $\gamma_u$ in Eq. (\ref{fu}), which is balanced. We then immediately deduce the following Corollary.

\begin{Cor}[Lowest possible absolute indicator for APN functions]\label{prop:APN}
    Let $n$ be a positive integer.
    If there exists an APN function from $\gf_2^n$ to $\gf_2^n$ with absolute indicator~$\Delta$, then there exists a balanced Boolean function of $n$~variables with linearity~$\Delta$.
\end{Cor}
  
To our best knowledge, the smallest known linearity for a balanced function is obtained by Dobbertin's recursive construction \cite{FSE:Dobbertin94}. For instance, for $n=9$, the smallest possible linearity for a balanced Boolean function is known to belong to the set $\{ 24, 28, 32 \}$, which implies that exhibiting an APN function over $\gf_{2}^9$ with absolute indicator $24$ would determine the smallest linearity for such a function.

One of the functions whose absolute indicator is known is the inverse
	mapping $F(x)=x^{2^n-2}$ over $\gf_{2^n}$.
	
	\begin{Prop} \label{Prop-Inv}\cite{Cha07}
	The autocorrelation spectrum of the inverse function  $F(x)=x^{2^n-2}$ over $\gf_{2^n}$ is given by
	$$\Lambda_F =  \left\{K\left(v \right) - 1 + 2 \times (-1)^{\tr_{2^n}(v)}: v\in \gf_{2^n}^{*} \right\},$$ where   
	$ K(a) = \sum_{x\in\gf_{2^n}^{*}}(-1)^{\tr_{2^n}\left(\frac{1}{x}+ax\right)} $ is the Kloosterman sum over $\gf_{2^n}$.
	Furthermore, the absolute indicator of the inverse function is given by:
	\begin{enumerate}[i)]
		\item when $n$ is even, $\Delta_{F}=2^{\frac{n}{2}+1}$;
		\item when $n$ is odd, $\Delta_{F} = \mathtt{L}(F)$ if $\mathtt{L}(F) \equiv 0 \pmod 8$, and $\Delta_{F} = \mathtt{L}(F)\pm 4$ otherwise.
	\end{enumerate}
	\end{Prop}

 When \(n\) is odd, the inverse mapping is APN. Then, from Proposition \ref{Prop_APN}, its autocorrelation table is directly determined by the corresponding $\gamma$.
This explains why the absolute indicator of the inverse mapping when \(n\) is odd, is derived from its linearity as detailed in the following example.
\begin{example}[ACT of the inverse mapping, $n$ odd]
    For any $u \in \gf_{2^n}^*$, the Boolean function $\gamma_u$, which characterizes the support of Row~$u$ in the DDT of the inverse mapping $F:x \mapsto x^{-1}$, coincides with $(1+F_{u^{-1}})$ except on two points:
    \begin{equation*}
        \gamma_u(x) = \begin{cases*}
            1+\tr(u^{-1}x^{-1}) & if $x \not \in \{0, u^{-1}\}$\\
            0                   & if $x = 0$\\
            1                   & if $x = u^{-1}$
        \end{cases*}\;.
    \end{equation*}
    This comes from the fact that the equation
    \begin{equation*}
        (x+u)^{-1} + x^{-1} = v
    \end{equation*}
    for $v \neq u^{-1}$ can be rewritten as
    \begin{equation*}
        x + (x+u) = v (x+u)x
    \end{equation*}
    or equivalently when $v \neq 0$, by setting $y=u^{-1}x$,
    \begin{equation*}
        y^2+y = u^{-1}v^{-1}\;.
    \end{equation*}
    It follows that this equation has two solutions if and only if $\tr_{2^n}(u^{-1}v^{-1}) = 0$.
    From the proof of the previous proposition, we deduce
    \begin{align*}
        \AC_{F}(u,v)
        &= -W_{\gamma_u}(v)  \\
        &= W_{F_{u^{-1}}}(v) + 2\left(1 - (-1)^{\tr_{2^n}(u^{-1}v)}\right)\;,
    \end{align*}
    where the additional term corresponds to the value of the sum defining the Walsh transform $W_{F_{u^{-1}}}(v)$ at points~$0$ and $u^{-1}$.
\end{example}

\section{Autocorrelation spectra and absolute indicator of special polynomials}
\label{polynomials}

This section mainly considers some polynomials of special forms. Explicitly, we will investigate the autocorrelation spectra and the absolute indicator of the Gold permutations and their inverses, and of the Bracken-Leander functions. Our study is divided into two subsections.

\subsection{Monomials}

In the subsection, we consider the autocorrelation of some special monomials of cryptographic interest, mainly APN permutations and one permutation with differential uniformity $4$, over the finite field $\gf_{2^n}$. 
Firstly, we present a general observation on the autocorrelation of monomials.

\begin{Prop}
	\label{prop_monomials}
	Let $F(x) = x^d \in \gf_{2^n}[x]$. Then $$\Lambda_F = \left\{ \AC_F(1,v): v\in\gf_{2^n}^{*} \right\}.$$ Moreover, if $\gcd\left(d,2^n-1\right)=1$, then 
	$$\Lambda_F = \left\{ \AC_F(u,1): u \in \gf_{2^n}^{*}  \right\}.$$
\end{Prop}

\begin{proof}
	For any $u,  v\in\gf_{2^n}^{*},$ we have
	\begin{eqnarray*}
		\AC_F(u,v) &=& \sum_{x\in\gf_{2^n}}(-1)^{\tr_{2^n}\left( v(F(x)+F(x+u))\right)  }\\
		&=& \sum_{x\in\gf_2^n}(-1)^{ \tr_{2^n}\left( v\left(  x^d+(x+u)^d \right) \right) }\\
		&=& \sum_{x\in\gf_{2^n}}(-1)^{\tr_{2^n}\left(vu^d\left( \left(\frac{x}{u}\right)^d +   \left(\frac{x}{u}+1\right)^d \right)   \right)}\\
		&=& \AC_F\left(1, vu^d \right).
	\end{eqnarray*}
	Moreover, if $\gcd\left( d,2^n-1 \right)=1$, then for any $v\in\gf_{2^n}^{*}$, there exists a unique element $u\in\gf_{2^n}^{*}$ such that $v=u^d$. Furthermore, 
	\begin{eqnarray*}
		\AC_F(1,v) &=& \sum_{x\in\gf_{2^n}}(-1)^{\tr_{2^n}\left(v\left(x^d+(x+1)^d\right)\right)} \\
		&=& \sum_{x\in\gf_{2^n}}(-1)^{\tr_{2^n}\left((ux)^d+(ux+u)^d\right)}\\
		&=& \sum_{y\in\gf_{2^n}}(-1)^{\tr_{2^n}\left(y^d+(y+u)^d\right)}\\
		&=& \AC_F(u,1).
	\end{eqnarray*}
	\end{proof}

Proposition \ref{prop_monomials} implies that it suffices to focus on the autocorrelation
	of the single component function $\tr_{2^n}\left(x^d\right)$ in the study of the autocorrelation table of the monomial $x^d$ with 
$\gcd\left(d,2^n-1\right)=1$. 	



We next discuss the autocorrelation of some cubic monomials. 
From Proposition \ref{prop-cubic}, if $n=m$ is odd, we obviously have that $\Delta_F \geq 2^{\frac{n+1}{2}}$. Furthermore, the equality is achieved when $ \dim(\left\{w\in \gf_2^n ~|~ D_{u}D_{w}f_{v} = c\right\})=1$ for all nonzero $u$ and $v$. Additionally, an upper bound on the absolute indicator can be established for two cubic APN permutations, namely the Kasami power function and the Welch function.
We denote the Kasami power functions $K_i$ and the Welch power function $W$ by
\[\begin{array}{rccl}
K_i : & \gf_{2^n} &\to & \gf_{2^n} \\
& x &\mapsto & x^{2^{2i}-2^i+1}
\end{array}
\qquad \text{and} \qquad
\begin{array}{rccl}
W : & \gf_{2^n} &\to & \gf_{2^n} \\
& x &\mapsto & x^{2^{(n-1)/2}+3}\;.
\end{array}\]

\begin{Prop}\cite[Lemma 1]{TIT:Carlet08}
	The absolute indicator for $W$ on $\gf_{2^n}$ is upper bounded by
	\begin{equation*}
	\Delta_W \leq 2^{\frac{n+5}{2}}.
	\end{equation*}
\end{Prop}

As long as the (regular) degree of the derivatives is small compared to the field size, the Weil bound gives a nontrivial upper bound for the absolute indicator of a vectorial Boolean function.
This is particularly interesting for the Kasami functions as the Kasami exponents do not depend on the field size (contrary to for example the Welch exponent).

\begin{Prop}
	The absolute indicator of $K_i$ on $\gf_{2^n}$ is upper bounded by
	\begin{equation*}
	\Delta_{K_i} \leq (4^i-2^{i+1})\times 2^{\frac{n}{2}}.
	\end{equation*}
	In particular,
	\begin{equation*}
	\Delta_{K_2} \leq 2^{\frac{n+5}{2}}.
	\end{equation*}
\end{Prop}
\begin{proof}
	Note that the two exponents with the highest degree of any derivative of $K_i$ are $4^i-2^i$ and $4^i-2^{i+1}+1$.
	The first exponent is even, so it can be reduced using the relation $\tr_{2^n}(y^2)=\tr_{2^n}(y)$.
	The result then follows from the Weil bound.
	Combining the bound with Proposition \ref{prop-cubic} yields the bound on $K_2$.
\end{proof}
%

Some other results on the autocorrelations of cubic Boolean functions $\tr_{2^n}(x^d)$ are known in the literature, which can be trivially extended to the vectorial functions $x^d$ if $\gcd(d,n)=1$, see~\cite[Theorem 5]{SAC:GonKho03}, \cite{TIT:Carlet08} and \cite[Lemmas 2 and 3]{IS:SunWu09}.
In the case $n=6r$ and $d=2^{2r}+2^r+1$, the power monomial $x^d$ is not a permutation, but results for all component functions of $x^d$ were derived in \cite{Can08}.
We summarize these results about the absolute indicator in the following proposition.

\begin{Prop}
	Let $F(x)=x^d$ be a function on $\gf_{2^n}$.
	\begin{enumerate}
		\item If $n$ is odd and $d=2^r+3$ with $r=\frac{n+1}{2}$, then $\Delta_{F}\in \{2^\frac{n+1}{2},2^\frac{n+3}{2}\}$.
		\item If $n$ is odd and $d$ is the $i$-th Kasami exponent, where $3i \equiv \pm 1 \pmod n$, then $\Delta_{F}=2^\frac{n+1}{2}$.
		\item If $n=2m$ and $d=2^{m+1}+3$, then $\Delta_{F}\leq 2^{\frac{3m}{2}+1}$.
		\item If $n=2m$, $m$ odd and $d=2^{m}+2^{\frac{m+1}{2}}+1$, then $\Delta_{F}\leq 2^{\frac{3m}{2}+1}$.
		\item If $n=6r$ and $d=2^{2r}+2^r+1$, then $\Delta_{F}=2^{5r}$.
	\end{enumerate}
\end{Prop}

 We now provide a different proof of the second case in the previous proposition that additionally relates the autocorrelation table of $K_i$ with the Walsh spectrum of a Gold function.
 
 \begin{Prop}\cite{Dillon}\label{thm:support}
 	Let $n$ be odd, not divisible by $3$ and $3i \equiv \pm 1 \pmod n$.
 	Set $f=\tr_{2^n}(x^d)$ where $d=4^i-2^i+1$ is the $i$-th Kasami exponent.
 	Then
 	\begin{equation*}
 	\supp(W_{f}) = \left\{x ~|~ \tr_{2^n}(x^{2^i+1})=1\right\}.
 	\end{equation*}
 \end{Prop}

\begin{Prop}\label{prop:kasami}
	Let $n$ be odd, not divisible by $3$ and $3i \equiv \pm 1 \pmod n$.
	Then
	\begin{equation*}
	\AC_{K_i}(u,v) = -\sum_{x \in \gf_{2^n}} (-1)^{\tr_{2^n}(uv^{1/d}x+x^{2^i+1})},
	\end{equation*}
	where $d = 4^i-2^i+1$ is the $i$-th Kasami exponent and $1/d$ denotes the inverse of $d$ in $\Z_{2^n-1}$.
	In particular, $\Delta_{K_i}=2^\frac{n+1}{2}$.
\end{Prop}
\begin{proof}
	It is well-known that, if $F$ is a power permutation over a finite field, its Walsh spectrum is uniquely defined by the entries $W_{F}(1,b)$.
	Indeed, for $v \neq 0$,
	\begin{align*}
	W_{K_i}(u,v)&=\sum_{x \in \gf_{2^n}} (-1)^{\tr_{2^n} (ux + v x^d )} 
	= \sum_{x \in \gf_{2^n}} (-1)^{\tr_{2^n} (uv^{-1/d} x + x^d )} = W_{K_i}(uv^{-1/d}, 1). 
	\end{align*}
Define a Boolean function 
	$$\widetilde{f}_{v}(x) = \begin{cases}
	1, & \text{ if } W_{K_i}(x,v) \neq 0 \\
	0, & \text{ if } W_{K_i}(x,v)=0. \\
	\end{cases}$$
	By Proposition \ref{thm:support}, the function $\widetilde{f}_{v}$ becomes 
		$$
		\widetilde{f}_{v}(x)  = \tr_{2^n}((v^{-1/d} x)^{2^i+1}).
		$$
	It follows from Proposition \ref{Prop_AB} that, for any $u$ and $v$,

$$
	\AC_{K_i}(u,v) =-W_{\tilde{f}_v}(u) = -\sum_{x \in \gf_{2^n}} (-1)^{\tr_{2^n}(ux+(v^{-1/d} x)^{2^i+1})}=-\sum_{x \in \gf_{2^n}} (-1)^{\tr_{2^n}(uv^{1/d}x+x^{2^i+1})}.$$
	Observe that $\gcd(i,n)=1$, so the Gold function $x^{2^i+1}$ is AB and
	$
	\AC_{K_i} = 2^{\frac{n+1}{2}}\;.
	$
\end{proof}
Note that the cases $3i \equiv 1 \pmod n$ and $3i \equiv -1 \pmod n$ are essentially only one case because the $i$-th and $(n-i)$-th Kasami exponents belong to the same cyclotomic coset.
Indeed, $(4^{(n-i)}-2^{n-i}+1)2^{2i} \equiv 4^{i}-2^i+1 \pmod {2^n-1}.$


%

From the known result in the literature, it appears that $(n, n)$-functions with a low absolute indicator are rare objects, which is also confirmed by experimental results for small integer $n$.  Below we propose an open problem for such functions.
\begin{Prob}
	For an odd integer $n$, are there power functions $F$ over $\gf_{2^n}$ with $\Delta_F=2^{(n+1)/2}$ other than the Kasami APN functions?
\end{Prob}

The Bracken-Leander function \cite{BL2010} is also a cubic permutation with differential uniformity $4$. 
In the following, we determine the autocorrelation spectrum and the absolute indicator of the Bracken-Leander function.

\begin{Th}
	Let $F(x) = x^{q^2+q+1}\in\gf_{q^4}[x]$, where $q=2^k$. Then for any nonzero $u, v$,
	$$\AC_F(u,v) \in \left\{  -q^3, 0, q^3  \right\} $$ and $\Delta_F = q^3$. 
\end{Th}

\begin{proof}
	For any $v\in\gf_{q^4}^{*}$,
	\begin{eqnarray*}
		\AC_F(1,v) &=& \sum_{x\in\gf_{q^4}}(-1)^{\tr_{q^4} \left( v(F(x)+F(x+1) ) \right) } \\
		&=& \sum_{x\in\gf_{q^4}}(-1)^{\tr_{q^4} \left( v\left( x^{q^2+q}+x^{q^2+1}+x^{q+1}+x^{q^2}+x^{q}+x+1 \right)  \right) }\\
		&=& (-1)^{\tr_{q^4}(v)} \sum_{x\in\gf_{q^4}}(-1)^{\tr_{q^4}\left( vx^{q^2+1}+\left( v^{q^3}+v \right)x^{q+1} +\left( v^{q^3}+v^{q^2}+v  \right) x   \right) }
	\end{eqnarray*}
	Moreover, 
	\begin{eqnarray*}
		\AC_F(1,v)^2 
		&=&  \sum_{x,y\in\gf_{q^4}}(-1)^{\tr_{q^4}\left( vx^{q^2+1}+\left( v^{q^3}+v \right)x^{q+1} +\left( v^{q^3}+v^{q^2}+v  \right) x +vy^{q^2+1}+\left( v^{q^3}+v \right)y^{q+1} +\left( v^{q^3}+v^{q^2}+v  \right) y   \right) } \\
		&= &  \sum_{x,y\in\gf_{q^4}}(-1)^{\tr_{q^4}\left( v(x+y)^{q^2+1}+\left( v^{q^3}+v \right)(x+y)^{q+1} +\left( v^{q^3}+v^{q^2}+v  \right) (x+y) +vy^{q^2+1}+\left( v^{q^3}+v \right)y^{q+1} +\left( v^{q^3}+v^{q^2}+v  \right) y   \right) } \\
		&=& \sum_{x,y\in\gf_{q^4}}(-1)^{\tr_{q^4}\left( v\left( x^{q^2+1}+xy^{q^2}+x^{q^2}y \right) + \left( v^{q^3}+v \right) \left(x^{q+1}+xy^{q}+x^qy  \right)  + \left( v^{q^3}+v^{q^2}+v \right) x \right) } \\
		&=&\sum_{x\in\gf_{q^4}}(-1)^{  \tr_{q^4} \left(  vx^{q^2+1} +\left(v^{q^3}+v \right)x^{q+1}+\left( v^{q^3}+v^{q^2}+v \right) x  \right)   }\sum_{y\in\gf_{q^4}}(-1)^{\tr_{q^4}(L_v(x)y)},
	\end{eqnarray*}	 
	where $L_v(x) = \left( v^{q^3}+v^{q^2}\right)x^{q^3}+\left( v^{q^2}+v \right)x^{q^2}+\left( v^{q^3}+v\right)x^q.  $ Let 
	$\ker\left(L_v \right) := \left\{  x\in\gf_{q^4} | L_v(x) =0 \right\}.$ Then 
	$$\AC_F(1,v)^2  = q^4 \times \sum_{x\in\ker\left(L_v \right)} (-1)^{ \phi_v(x) },  $$ where $\phi_v(x) = \tr_{q^4} \left(  vx^{q^2+1} +\left(v^{q^3}+v \right)x^{q+1}+\left( v^{q^3}+v^{q^2}+v \right) x  \right). $  
	
	(1) When $v\in\gf_q^{*}$, $L_v(x)=0$ and thus $\ker\left(L_v \right) = \gf_{q^4}$. Moreover, $ \phi_v(x) = \tr_{q^4}\left(  vx^{q^2+1} +vx \right) = \tr_{q^4}\left(vx\right). $ Therefore, $$\AC_F(1,v)^2  = q^4\times \sum_{x\in \gf_{q^4}} (-1)^{ \tr_{q^4}(vx) } = 0.$$

	(2) When $v\in\gf_{q^4}\backslash\gf_{q}$, $\phi_v$ is linear on $\ker\left(L_v \right)$, which can be proved by direct computations. Thus $\AC_F(1,v)^2\neq 0$ only when $\phi_v$ is the all-zero mapping on $\ker\left(L_v \right)$. In addition, there must exist some $v$ such that $\AC_F(1,v)\neq0$ since $F$ is not bent. Moreover, the Dickson matrix of $L_v$ is 
	$$D = \begin{pmatrix}
	0 & v^{q^3}+v &  v^{q^2}+v & v^{q^3}+v^{q^2}\\
	v^{q^3}+v & 0 & v^q+v & v^{q^3}+v^q \\
	v^{q^2}+v & v^q+v & 0 & v^{q^2}+v^q \\
	v^{q^3}+v^{q^2} & v^{q^3}+v^q & v^{q^2}+v^q & 0 
	\end{pmatrix}.$$
	It is easy to compute that the rank of $D$ is $2$ and thus $\#\ker\left(L_v\right)=q^2$. Therefore, there exists some $v$ with $$\AC_F(1,v)^2 = q^4\sum_{x\in\ker\left(L_v \right)} (-1)^{ \phi_v(x)} =q^4 \#\ker\left(L_v\right) = q^6. $$
	This completes the proof.
\end{proof}

\subsection{Quadratic functions and their inverses}
\label{sec:quadratic}

In this subsection,  we firstly consider the general quadratic functions and determine the autocorrelation spectra of the Gold functions and of their inverses.

\begin{Th}
	\label{quad_func}
	Let $F(x) = \sum_{0\le i<j\le n-1}a_{ij}x^{2^i+2^j}\in\gf_{2^n}[x]$. Then the autocorrelation table of $F$ takes values from $\{0, \pm 2^{n}\}$
	and $\Delta_F = 2^{n}.  $
\end{Th}

\begin{proof}
	For any  $u,  v\in\gf_{2^n}^{*},$
	\begin{eqnarray*}
		\AC_F(u,v) &=& \sum_{x\in\gf_{2^n}}(-1)^{\tr_{2^n}\left( v(F(x)+F(x+u))\right)  }\\
		&=&\sum_{x\in\gf_{2^n}}(-1)^{\tr_{2^n}\left( v\left( \sum_{0\le i<j\le n-1}  a_{ij}\left(  u^{2^j}x^{2^i} +  u^{2^i} x^{2^j}+u^{2^i+2^j}  \right) \right)  \right)}\\
		&=&(-1)^{\tr_{2^n}\left( v \left(  \sum_{0\le i<j\le n-1}  a_{ij} u^{2^i+2^j} \right)  \right)} \sum_{x\in\gf_{2^n}} (-1)^{\tr_{2^n}  \left(   L(u,v) x\right)}, \\
	\end{eqnarray*}
where $L(u,v) =  \sum_{0\le i<j\le n-1}\left( a_{ij}^{2^{-i}}u^{2^{j-i}}v^{2^{-i}} + a_{ij}^{2^{-j}} u^{2^{i-j}} v^{2^{-j}} \right).$	When $ L(u,v) = 0,$ $\AC_F(u,v) = \pm 2^n$; otherwise, $\AC_F(u,v) = 0.$ Thus $\AC_F(u,v) \in \left\{ -2^{n}, 0, 2^{n} \right\}. $ Moreover, since $F$ cannot be bent, we obtain $\Delta_F\neq0$ and then $\Delta_F = 2^{n}$.
\end{proof}

\begin{Cor}
	Let $F(x) = x^{2^i+1}\in\gf_{2^n}[x]$. Assume $d=\gcd(i,n)$ and $n'=n/d$. Then
		\begin{equation*}
	\Lambda_F =	\left\{
	\begin{array}{lr}
	\{ 0, 2^{n} \},  &~ \text{if}~ n^{'}~ \text{is even}, \\
	\{ -2^{n}, 0 \}, &~ \text{if}~ n^{'} ~\text{is odd and}~ d=1,\\
	\{ -2^{n}, 0, 2^{n} \}, &~ \text{otherwise.}
	\end{array}
	\right.
	\end{equation*}   
\end{Cor}

\begin{proof}
	From the proof of Theorem \ref{quad_func}, it is clear that 
	$$\AC_{F}(1,v) = (-1)^{\tr_{2^n}(v)} \sum_{x\in\gf_{2^n}} (-1)^{\tr_{2^n}(L(v)x)}, $$
	where $L(v) = v^{2^{-i}}+v.$ Thus $\ker( L) = \gf_{2^{\gcd(i,n)}} = \gf_{2^d}$. Furthermore, for any $v\in\gf_{2^d}$, $\tr_{2^n}(v)=n^{'}\tr_{2^d}(v)$. Therefore, 
	\begin{equation*}
 \AC_F(1,v) =	\left\{
	\begin{array}{lr}
	0,  &~  \text{if}~ v \in\gf_{2}^n\backslash\gf_{2}^d, \\
	2^n\times (-1)^{n^{'}\tr_{2^d}(v)}, &~ \text{if}~ v\in\gf_{2}^d.
	\end{array}
	\right.
	\end{equation*}
   It follows that
	\begin{equation*}
\Lambda_F =	\left\{
\begin{array}{lr}
\{ 0, 2^{n} \}, &~ \text{if}~ n^{'}~ \text{is even}, \\
\{ -2^{n}, 0 \}, &~ \text{if}~ n^{'} ~\text{is odd and}~ d=1,\\
\{ -2^{n}, 0, 2^{n} \}, &~ \text{otherwise.}
\end{array}
\right.
\end{equation*}   
\end{proof}

As previously observed, the autocorrelation spectrum and the absolute indicator are not invariant under compositional inversion. Then, in the following, we consider the absolute indicator of the inverse of a quadratic permutation, which is not obvious at all. Indeed, the absolute indicator depends on the considered function, as we will see next.

For example, for $n=9$, the inverses of the two APN Gold permutations $x^3$ and $x^5$, namely $x^{341}$ and $x^{409}$, do not have the same absolute indicator: the absolute indicator of $x^{341}$ is $56$ while the absolute indicator of $x^{409}$ is $72$. 

Nevertheless, the specificity of quadratic APN permutations for $n$ odd is that they are crooked \cite{BF98}, which means that the image set of every derivative $D_uF, u \neq 0$, is the complement of a hyperplane $ \langle \pi(u) \rangle^{\perp}.$ Moreover, it is known (see e.g. \cite[Proof of Lemma 5]{CC03}) that all these hyperplanes are distinct, which implies that $\pi$ is a permutation of $\gf_{2}^n$ when we add to the definition that $\pi(0)=0$. Then, the following proposition shows that, for any quadratic APN permutation $F$, the autocorrelation of $F^{-1}$ corresponds to the Walsh transform of~$\pi$.

\begin{Prop}
	\label{pi}
	Let $n$ be an odd integer and $F$ be a quadratic APN permutation over $\gf_2^n$. Let further $\pi$ be the permutation of $\gf_2^n$ defined by 
	$$ \mathsf{Im} (D_uF) = \gf_{2}^n \backslash \langle \pi(u) \rangle^{\perp}, ~~\text{when}~~ u \neq 0, $$
	and $\pi(0)=0$. Then for any nonzero $u,v$ in $\gf_2^n$, we have 
	$$\AC_{F^{-1}}(u,v) = -W_{\pi}(v,u). $$
	It follows that $$\Delta_{F^{-1}}\ge 2^{\frac{n+1}{2}}$$
	with equality if and only if $\pi$ is an AB permutation.
\end{Prop}

\begin{proof}
	Let $u,v$ be two nonzero elements of $\gf_2^n$. Then, from Eq. (\ref{DLCT_1_DDT}), we deduce
	\begin{eqnarray*}
	\AC_{F^{-1}}(u,v) &=& \sum_{\omega\in\gf_{2}^m}(-1)^{v\cdot\omega} \mathtt{DDT}_{F^{-1}}(u,\omega) \\
	&=& \sum_{\omega\in\gf_{2}^m}(-1)^{v\cdot\omega} \mathtt{DDT}_{F}(\omega,u).
	\end{eqnarray*}
By the definition of $\pi$, we have that, for any nonzero $a$, 
$$\mathtt{DDT}_F(a,b) = 
\left\{
\begin{array}{ll}
2, & ~~\text{if}~~ b\cdot \pi(a) =1, \\
0, &~~\text{if}~~ b\cdot \pi(a) =0.
\end{array}
\right.$$
It then follows that
$$\mathtt{DDT}_F(a,b) = 1-(-1)^{\pi(a)\cdot b},$$
where this equality holds for all $(a,b)\neq (0,0)$ by using that $\pi(0)=0$. Therefore, we have, for any nonzero $u$ and $v$,
$$ \AC_{F^{-1}}(u,v) = - \sum_{\omega\in\gf_{2}^m}(-1)^{v\cdot\omega} \left(1-(-1)^{\pi(\omega)\cdot u}  \right) = -W_{\pi}(v,u). $$

As a consequence, $\Delta_{F^{-1}}$ is equal to the linearity of $\pi$, which is at least $2^{\frac{n+1}{2}}$ with equality for AB functions.
\end{proof}

It is worth noticing that the previous proposition is valid, not only for quadratic APN
permutations, but for all \textit{crooked} permutations, which are a particular case of AB functions. However, the existence of crooked permutations of degree strictly higher than 2 is an open question.

As a corollary of the previous proposition, we get some more precise information on
the autocorrelation spectrum of the quadratic power permutations corresponding to the inverses of the Gold
functions. Recall that $x^{2^i+1}$ and $x^{2^{n-i}+1}$ are affine equivalent since the two exponents belong to the same cyclotomic coset modulo $(2^n-1)$.  This implies that their inverses share the same autocorrelation spectrum.

\begin{Cor}
	Let $n>5$ be an odd integer and $0<i<n$ with $\gcd(i,n)=1$. Let $F$ be the APN power permutation over $\gf_{2^n}$ defined by $F(x)=x^{2^i+1}$. Then, for any nonzero $u$ and  $v$ in $\gf_{2^n}$, we have 
	$$  \AC_{F^{-1}}(u,v) = - W_{\pi}(v,u), ~~\text{where}~~ \pi(x) = x^{2^n-2^i-2}.  $$
	Most notably, the absolute indicator of $F^{-1}$ is strictly higher than $2^{\frac{n+1}{2}}$.
\end{Cor}

\begin{proof}
The result comes from the form of the function $\pi$ which defines the DDT of $x^{2^i+1}$.
Indeed, for any nonzero $u\in\gf_{2^n}$ the number $\mathtt{DDT}_F(u,v)$ of solutions of
$$(x+u)^{2^i+1}+x^{2^i+1}=v$$
is equal to the number of solutions of 
$$x^{2^i}+x=1+vu^{-(2^i+1)},$$
which is nonzero if and only if $\tr_{2^n}\left(vu^{-(2^i+1)} \right) = 1.$ It follows that 
$$\pi(x) = x^{2^n-2^i-2}.$$
Then the autocorrelation of $F^{-1}$ then follows from Proposition \ref{pi}. Moreover, this function $\pi$ cannot be AB since AB functions have algebraic degree at most $\frac{n+1}{2}$ \cite[Theorem 1]{CCZ98}, while $\pi$ has degree $(n-2)$. It follows that $\pi$ cannot be AB when $n>5$. Therefore, 	the absolute indeed of $F^{-1}$ is strictly higher than $2^{\frac{n+1}{2}}$.
\end{proof}

In the specific case $n=5$, it can easily be checked that the inverses of all Gold APN
permutations $F(x)=x^{2^i+1}$ have absolute indicator~$8$.

\section{Conclusion}
\label{conclusion}
{This paper intensively investigates the differential-linear connectivity table (DLCT) of vectorial Boolean functions.
The main contributions of this paper are four-fold.
Firstly, we reveal the connection between DLCT and the autocorrelation table of vectorial Boolean functions 
and we characterize these two notions in terms of the Walsh transform of the function and of its differential distribution table.
Secondly, we provide bounds on the absolute indicator of \((n,m)\)-functions when \(m \geq n\) and we exhibit the divisibility property of the autocorrelation of any vectorial Boolean functions. Moreover, we
investigate the invariance of the autocorrelation table under affine, EA and CCZ equivalence and exhaust the autocorrelation spectra of optimal $4$-bit S-boxes. 
Thirdly, we analyze some properties of the autocorrelation of cryptographically desirable functions, including  APN, plateaued and AB functions and express the autocorrelation of APN and AB functions with the Walsh transform of certain Boolean functions. Finally, we investigate the autocorrelation spectra of some special polynomials, including monomials with low differential uniformity, cubic monomials, quadratic functions and inverses of quadratic permutations. 

This paper only covers a small portion of interesting problems on this subject and many problems deserve further research. For instance,  
the generic lower bound on the absolute indicator of vectorial Boolean functions derived in this paper is lower than what experimental results suggest and thus might be further improved. A natural follow-up topic would be the investigation and construction of optimal, or near-optimal, vectorial Boolean functions with respect to the bounds.
}

{\bfseries Note: } The current paper is a merged version of  \cite{CKW2019} and \cite{LLLQ2019}.

\bibliographystyle{plain}

\bibliography{ref}

\end{document}